\documentclass[12pt,a4paper]{article}

\usepackage{amsfonts,amsmath,amssymb,amsthm}
\usepackage{latexsym,amscd}
\usepackage{amsbsy}
\usepackage{CJK}
\usepackage{indentfirst,mathrsfs}
\usepackage{amsfonts,amsmath,amssymb,amsthm}
\usepackage{latexsym,amscd}
\usepackage{amsbsy}
\usepackage{color}
\usepackage{multirow}
\usepackage{makecell}
\usepackage{float}
\usepackage{cases}

\numberwithin{equation}{section}

\newtheorem{thm}{Theorem}
\newtheorem{lem}{Lemma}

\newtheorem{remark}{Remark}

\newcommand{\ftwon}{{\mathbb F}_{2^n}}

\newcommand{\ftwom}{{\mathbb F}_{2^m}}

\newcommand{\ftwo}{{\mathbb F}_{2}}
\newcommand{\Tr}{{\rm {Tr}}}

\begin{document}

\title{$4$-uniform BCT permutations from generalized butterfly structure}
\author{ Nian Li, Zhao Hu
\thanks{N. Li, Z. Hu and X. Zeng are at the Hubei Key Laboratory of Applied Mathematics, Faculty of Mathematics and Statistics, Hubei University, Wuhan, 430062, China, and also with the State Key Laboratory of Cryptology, P.O. Box 5159, Beijing 100878, China.  Email: nian.li@hubu.edu.cn, zhao.hu@aliyun.com, xzeng@hubu.edu.cn}
, Maosheng Xiong
\thanks{M. Xiong is at the Department of Mathematics, The Hong Kong University of Science and Technology,
Clear Water Bay, Kowloon, Hong Kong, China. E-mail: mamsxiong@ust.hk}
and Xiangyong Zeng
}
\date{}
\maketitle

\begin{quote}
  {\small {\bf Abstract:}
As a generalization of Dillon's APN permutation, butterfly structure and generalizations have been of great interest since they generate permutations with the best known differential and nonlinear properties over the field of size $2^{4k+2}$. Complementary to these results, we show in this paper that butterfly structure, more precisely the closed butterfly also yields permutations with the best boomerang uniformity, a new and important parameter related to boomerang-style attacks. This is the sixth known infinite family of permutations in the literature with the best known boomerang uniformity over such fields. }

\end{quote}

\section{Introduction}

Substitution boxes (S-boxes) are important components of block ciphers. Being the only source of nonlinearity in these ciphers, they play a central role in the robustness by obscuring the relationship between the key and ciphertext. The security of most modern block ciphers relies significantly on the cryptographic properties of their S-boxes. It is crucial to employ S-boxes with good cryptographic properties in order to resist various attacks \cite{BSD,Lai,M}.

Mathematically, S-boxes are vectorial (multi-output) Boolean functions, that is, functions $F: V \to V'$ where $V$ and $V'$ are $m$ and $n$-dimensional vector spaces over the binary field $\ftwo$ respectively. For $a,b\in V=V'$, the entries of the DDT (Difference Distribution Table) of $F$ are given by
$${\rm DDT}_{F}(a,b)=\#\left\{x\in V: F(x)+F(x+a)=b \right\}.$$
The differential uniformity of $F$ is defined as
$$\delta(F)=\max_{a\in V \setminus \{\mathbf{0}\}, b\in V}{\rm DDT}_{F}(a,b).$$
Differential uniformity is an important concept in cryptography as it quantifies the degree of security of the cipher with respect to differential attacks \cite{BSD} if $F$ is used as S-boxes in the cipher. In particular, if $\delta(F) =2$, then $F$ is called almost perfect nonlinear (APN), which offers maximal resistance to differential attacks if used as S-boxes.

Proposed by Wagner \cite{WAG} in $1999$, boomerang attack is an also important cryptanalysis technique against block ciphers involving S-boxes. It can be considered as an extension of the classical differential attack \cite{BSD} as in boomerang attack, two differentials are combined and analyzed for the upper and the lower parts of the cipher \cite{BDK01,BDK02,BDD03,BK09,DKS10,KKS01,KHP+12}. Shedding new light on the effectiveness of boomerang attack, in Eurocrypt $2018$, Cid et al. \cite{CHP} introduced a new cryptanalysis tool: Boomerang Connectivity Table (BCT). Let $F: V \to V$ be a permutation. The entries of the BCT are given by
$${\rm BCT}_{F}(a,b)=\# \left\{x\in V: F^{-1}(F(x)+b)+F^{-1}(F(x+a)+b)=a \right\},$$
where $F^{-1}$ denotes the compositional inverse of $F$. The boomerang uniformity of $F$, introduced by Boura and Canteaut in \cite{BCO}, is defined as
$$\beta(F)=\max_{a,b\in V \setminus \{\mathbf{0}\} }{\rm BCT}_{F}(a,b).$$ The function $F$ is called a $\beta(F)$-uniform BCT function.

In principle, the smaller the quantity $\beta(F)$, the stronger the security of the cipher against boomerang-style attacks if $F$ is used as S-boxes in the cipher. It was known in \cite{CHP} that $\beta(F) \ge \delta(F)$, and if $\delta(F)=2$, then $\beta(F)=2$, hence APN permutations offer maximal security against both differential attack and boomerang attack. However, in even dimension which is the most interesting in cryptography, Dillon's permutation in dimension $6$ (see \cite{BDMW}) remains the only known example of APN permutations. The existence of APN permutations in even dimension $\geq 8$ remains an open problem (This is the famous Big APN Problem \cite{BDMW}). Therefore, in even dimension, permutations $F$ with $\beta(F)=4$ offer the best resistance to differential and boomerang attacks.

It has been noted in \cite{CHP} that finding permutations in even dimension with $4$-uniform BCT is a hard problem, especially when the dimension increases. Up to now, there are only $5$ infinite families of permutations $F(x):\mathbb{F}_{2^n} \to \mathbb{F}_{2^n}$ with $4$-uniform BCT where $n \equiv 2 \pmod{4}$, which we list as below:
\begin{enumerate}
  \item [1)]   $F(x)=x^{2^n-2}$ (\cite{BCO});
  \item [2)]  $F(x)=x^{2^i+1}$, $\gcd(i,n)=2$ (\cite{BCO});
  \item [3)]  $F(x)=x^{2^t+2}+\gamma x$, $t=n/2$ and ${\rm ord}(\gamma^{2^t-1})=3$ (\cite{LQSL});
  \item [4)] $F(x)=\alpha x^{2^s+1}+\alpha^{2^{t}}x^{2^{-t}+2^{t+s}}$, $n=3t$, $t\equiv 2 \, ({\rm mod}\, 4)$, $\gcd(n,s)=2$, $3|(t+s)$ and $\alpha$ is a primitive element of $\mathbb{F}_{2^n}$ (\cite{MTX});
  \item [5)] $F(x)=x^{3\cdot2^{n/2}}+a_{1}x^{2^{(n+2)/2}+1}+a_{2}x^{2^{n/2}+2}+a_{3}x^3$, where $(a_{1},a_{2},a_{3})\in \Gamma_{1}$ (see \cite{TLZ} for details).
\end{enumerate}
The first two functions were proposed by Boura and Canteaut \cite{BCO}. Boura and Canteaut also showed that the boomerang uniformity is invariant up to affine equivalence and inversion and they entirely determined the value of the boomerang uniformity for all differentially $4$-uniform
permutations of $\mathbb{F}_{2^4}$. The third one is due to  Li, Qu, Sun and Li \cite{LQSL} who also presented an equivalent formula to compute the boomerang uniformity without knowing the compositional inverse of a permutation $F$. The fourth one is the Bracken-Tan-Tan's function \cite{BTT} and was shown to possess boomerang uniformity four by Mesnager, Tang and Xiong in \cite{MTX} where they used a slightly different formula to compute the boomerang uniformity and generalized earlier results on quadratic permutations with $4$-uniform BCT. The last one was recently presented by Tu, Li, Zeng and Zhou \cite{TLZ} by detailed study on solutions to a specific degree $2$ equation over finite fields.

The purpose of this paper is to present the sixth infinite family of permutations with $4$-uniform BCT in even dimension arising from generalized butterfly structure.

As a generalization of Dillon's APN permutation in dimension 6, butterfly structure was initially proposed by Perrin, Udovenko and Biryukov \cite{PUB} to generate $2m$-bit mappings by concatenating two bivariate functions over $\ftwom$. Canteaut, Duval and Perrin \cite{CDP} further studied this structure and generalized it as below. Let $R(x,y)$ be a bivariate polynomial on $\ftwom$ such that $R_{y}:x \mapsto R(x,y)$ is a permutation of $\ftwom$ for any $y \in \ftwom$. The \emph{closed butterfly} is the function $V_R: \ftwom\times\ftwom \to \ftwom\times\ftwom$ defined by
\begin{eqnarray} \label{1:close} V_{R}(x,y)=\left(R(x,y),R(y,x)\right),\end{eqnarray} and the \emph{open butterfly} is the function $H_{R}: \ftwom\times\ftwom \to \ftwom\times\ftwom$ defined by
\begin{eqnarray*} \label{1:open} H_{R}(x,y)=\left(R\left(y,R_{y}^{-1}(x)\right),R_{y}^{-1}(x)\right). \end{eqnarray*} Here $R_{y}(x):=R(x,y)$ and $R_{y}^{-1}$ is the compositional inverse of $R_{y}$, that is, $R_{y}^{-1}(R_{y}(x))=x$ for any $x,y \in \ftwom$. It is known that $H_{R}$ is always an involution (and hence a permutation) and the two functions $H_{R}$ and $V_{R}$ are CCZ-equivalent, so they share the same differential uniformity, nonlinearity and Walsh spectrum.

Let $m,k$ be positive integers such that $m$ is odd and $\gcd(k,m)=1$. Extending previous work \cite{CDP,FFW}, Li, Tian, Yu and Wang \cite{LTYW} considered a general bivariate polynomial $R(x,y)$ of the form
\begin{eqnarray*}
R(x,y)=(x+\alpha y)^{2^k+1}+\beta y^{2^k+1} \end{eqnarray*}
for any $\alpha,\beta \in \ftwom$ and proved that the corresponding butterflies $H_R$ and $V_R$ are differentially $4$-uniform and have the best known nonlinearity when $\beta \neq (\alpha+1)^{2^k+1}$. Under this condition, however, the closed butterfly $V_R$ may not be a permutation.

Since $\gcd(2^k+1,2^m-1)=1$, any nonzero $\beta \in \ftwom$ can be written as $\beta=\beta_1^{2^k+1}$ for some $\beta_1\in\ftwom$. So equivalently, for any $\alpha, \beta \in \ftwom$, we consider the general bivariate polynomial $R(x,y)$ of the form
\begin{eqnarray} \label{1:biva}
R(x,y)=\left(x+\alpha y\right)^{2^k+1}+\left(\beta y\right)^{2^k+1}. \end{eqnarray}
Our main result is as follows.
\begin{thm} \label{thm1}
Let $m,k$ be positive integers such that $m$ is odd and $\gcd(k,m)=1$. For any $\alpha, \beta \in\ftwom \setminus \ftwo$ such that
\begin{eqnarray}\label{ab}
\alpha^2+\beta^2+\alpha \beta+1 = 0,
\end{eqnarray}
let the bivariate polynomial $R(x,y)$ be defined in (\ref{1:biva}). Then the closed butterfly $V_{R}$ given in (\ref{1:close}) is a permutation on $\ftwom \times \ftwom$ with $4$-uniform BCT.
\end{thm}

\begin{remark}
Theorem \ref{thm1} complements previous work \cite{CDP,FFW,LTYW} on butterfly structure and gives rise to the sixth family of $4$-uniform BCT permutations in even dimension. It was known from \cite{LTYW} that the closed butterfly $V_R$ in Theorem \ref{thm1} possesses the best known nonlinearity when $\beta\ne\alpha+1$.
\end{remark}

\begin{remark}
Computer experiments indicate that Condition (\ref{ab}) is also necessary for $V_R$ to be a permutation. As to the open butterfly $H_R$, computer experiments indicate that the boomerang uniformity is much larger than $4$ in general. We may come back to these questions in the future.
\end{remark}

There are three ingredients in the proof of Theorem \ref{thm1}. First, we convert the vector-expression $V_R$ into a single univariate polynomial on $\mathbb{F}_{2^{2m}}$, which seems much easier to handle. Second, our computation rely crucially on a complete and explicit solvability criterion on solving a certain type of equations over finite fields (see Lemma \ref{lemma-core}). Such a criterion should be of independent interest for other applications. This is similar to \cite[Lemma 3]{TLZ} which played an essential role in the whole paper. Third, we apply the new formulation of the boomerang uniformity from \cite{LQSL} which is quite convenient for computations.

Other than the closed butterfly structure, it seems possible to treat directly univariate polynomials of general forms, similar to the ones studied in \cite{TLZ}, following the main techniques of this paper. We shall stress this problem in future work.

We organize this paper as follows: in Section \ref{2:pre} we collect some solvability criteria on certain equations over finite fields which will be used repeatedly in the paper; in Section \ref{sec3} we first show how to convert $V_R$ into a single univariate polynomial on $\mathbb{F}_{2^{2m}}$, and then discuss in details the solvability of the difference equation $F(x+a)+F(x)=b$; after this preparation, then in Section \ref{sec4} we prove the main results. To streamline the presentation, we postpone the proofs of some technical results to the end of the paper in Section {\bf Appendix}.

\section{Solving certain equations over finite fields} \label{2:pre}

The following three results will be used repeatedly in the rest of the paper.

\begin{lem}(\cite{Lidl}) \label{lem0-0}
Let $n$ be a positive integer. For any $a \in \ftwon^*:=\ftwon \setminus \{0\}$ and $b \in \ftwon$, the equation
\[x^{2}+ax+b=0\]
is solvable (with two solutions) in $\ftwon$ if and only if
\[\Tr_1^n\left(\frac{b}{a^2}\right)=0.\]
Here $\Tr_{1}^n$ is the absolute trace map from $\mathbb{F}_{2^n}$ to the binary field $\mathbb{F}_{2}$.
\end{lem}

\begin{lem}(\cite{Kim}) \label{lem0}
Let $n,k$ be positive integers such that $\gcd(n,k)=1$. For any $a \in \ftwon$, the equation
\[x^{2^k}+x=a\]
has either 0 or 2 solutions in $\ftwon$. Moreover, it is solvable with two solutions in $\ftwon$ if and only if $\Tr_1^n(a)=0$.
\end{lem}

\begin{lem} \label{lemma-core}
Let $m, k$ be odd integers such that $\gcd(k,m)=1$. Let $n=2m$. For any $\mu,\nu\in\ftwon$, define
\begin{equation*}
L_{\mu,\nu}(x)=x^{2^k}+\mu\overline{x}+(\mu+1)x+\nu.
\end{equation*}
Then the equation $L_{\mu,\nu}(x)=0$ has either $0$, $2$ or $4$ solutions in $\ftwon$. More precisely, let $\xi, \Delta\in\ftwom$ and $\lambda\in\ftwon$ be defined by the equations
\begin{eqnarray}\label{notation}
\xi^{2^k-1}=1+\mu+\overline{\mu},\;\;\; \Delta=\frac{\nu+\overline{\nu}}{\xi^{2^k}},\;\;\; \lambda^{2^k}+\lambda=\mu \xi.
\end{eqnarray}
Then
\begin{enumerate}
    \item [(1)] $L_{\mu,\nu}(x)=0$ has two solutions in $\ftwon$ if and only if one of the following conditions is satisfied:\\
         (i)  $1+\mu+\overline{\mu}=0$ and $\sum_{i=0}^{m-1}(\mu^{2^k}(\nu+\overline{\nu})+\nu^{2^k})^{2^{ki}}=\nu+\overline{\nu}$;\\
         (ii) $1+\mu+\overline{\mu}\neq 0$, $\Tr_{1}^{m}(\Delta)=0$ and $\overline{\lambda}+\lambda=\xi+1$.
    \item [(2)] $L_{\mu,\nu}(x)=0$ has four solutions in $\ftwon$ if and only if
     $1+\mu+\overline{\mu}\neq 0$, $\Tr_{1}^{m}(\Delta)=0$, $\overline{\lambda}+\lambda=\xi$ and
     $\Tr_{1}^{n}\left(\frac{\lambda^{2^k}\overline{\nu}}{\xi^{2^k}}\right)=0$.
\end{enumerate}
If $\nu=0, 1+\mu+\overline{\mu} \ne 0$ and $\lambda+\overline{\lambda}=\xi$, then the set of four solutions of $L_{\mu,\nu}(x)=0$ in $\ftwon$ is given by $\left\{0,1,\lambda, \lambda+1 \right\}$.
\end{lem}

\begin{proof}
See Appendix A.
\end{proof}

\begin{remark} When $k=1$, Lemma \ref{lemma-core} reduces to \cite[Lemma 3]{TLZ} which played a central role in computing the boomerang uniformity in the paper. Comparing with \cite[Lemma 3]{TLZ}, our criteria seems a little simpler.
\end{remark}

\section{Proofs of Theorem \ref{thm1}: in preparation} \label{sec3}

From the setting of Theorem \ref{thm1}, if $k$ is even, letting $k':=m-k$, then $k'$ is odd and $\gcd(k',m)=1$. Since
\[R(x,y)^{2^{k'}}= \left(x+\alpha y\right)^{2^{k'}+1}+\left(\beta y\right)^{2^{k'}+1}, \quad \alpha, \beta \in \ftwom, \]
the case of $k$ for Theorem \ref{thm1} is equivalent to the case of $k'$ which is odd. For this reason, from what follows we always assume that $m,k$ are odd integers and $\gcd(m,k)=1$.

\subsection{Univariate polynomial expression of $V_R$}

We now derive a univariate polynomial expression of $V_R$. Let $n=2m$ and $\omega$ be a root of $x^2+x+1=0$. Since $m$ is odd, $\{1,\omega\}$ is a basis of $\ftwon$ over $\ftwom$ and  $\ftwom^2$ is isomorphic to $\ftwon$ under the map
\[ z=(x,y)\mapsto x+\omega y, \qquad \forall x, y\in \ftwom.\]
Hence every element $z\in\ftwon$ can be uniquely represented as $z=x+\omega y$ with $x,y\in \ftwom$. This together with $\overline{z}=x+\overline{\omega}y$, where $\overline{z}:=z^{2^m}$, one obtains
\begin{eqnarray*}
 x=\overline{\omega}z+\omega\overline{z},\;\; y=z+\overline{z}.
\end{eqnarray*}
Substituting $z$ with $\omega^2 z$ gives
$$V_{R}(x,y)=V_{R}(z)=\omega^2(e_{1}z^{2^k+1}+e_{2}\overline{z}^{2^k+1}+e_{3}z^{2^k}\overline{z}+e_{4}z\overline{z}^{2^k}),$$
where
\begin{eqnarray*}
\begin{array}{llllll}
   e_{1}&=&1+\alpha+\alpha^{2^k+1}+\beta^{2^k+1},\;&e_{2}&=&1+\alpha^{2^{k}}+\alpha^{2^k+1}+\beta^{2^k+1}, \\
 e_{3}&=&1+\alpha+\alpha^{2^{k}},\;&e_{4}&=&\alpha+\alpha^{2^{k}}+\alpha^{2^k+1}+\beta^{2^k+1}.
\end{array}
\end{eqnarray*}
Thus, the closed butterfly $V_R$ defined by (\ref{1:close}) is affine equivalent to the polynomial
\begin{eqnarray} \label{2:uni}
e_{1}x^{2^k+1}+e_{2}\overline{x}^{2^k+1}+e_{3}x^{2^k}\overline{x}+e_{4}x\overline{x}^{2^k}.
\end{eqnarray}
Since $\alpha, \beta \in \ftwom \setminus \mathbb{F}_2$ satisfying $\alpha^2+\beta^2+\alpha \beta+1=0$, using $\beta=\theta\alpha+1$ for some $\theta\in\ftwom^*$, we find that a common solution of $(\alpha,\beta)$ is given by
$$(\alpha,\beta)=\left(\frac{1}{1+\theta+\theta^2},\frac{\theta^2}{1+\theta+\theta^2}\right), \quad \theta \in \ftwom^*.$$
Using the above expression, the quadrinomial (\ref{2:uni}) is affine equivalent to
\begin{eqnarray}\label{F}
F(x):=c_{1}x^{2^k+1}+c_{2}\overline{x}^{2^k+1}+c_{3}x^{2^k}\overline{x}+c_{4}x\overline{x}^{2^k},
\end{eqnarray}
where the coefficients $c_i=e_i\alpha^{-(2^k+1)}$ and are explicitly given by
\begin{eqnarray}\label{ci-theta}
\left\{\begin{array}{lll}
  c_{1}&=&1+\theta+\theta^2+(\theta+\theta^2)^{2^k+1}+\theta^{2(2^k+1)},   \\
  c_{2}&=&(1+\theta+\theta^2)^{2^k}+(\theta+\theta^2)^{2^k+1}+\theta^{2(2^k+1)},   \\
  c_{3}&=&1+(\theta+\theta^2)^{2^k+1},  \\
  c_{4}&=&(1+\theta+\theta^2)^{2^{k}+1}+(\theta+\theta^2)^{2^k+1}+\theta^{2(2^k+1)}.
\end{array} \right.
\end{eqnarray}
We conclude that the closed butterfly $V_{R}^{k}(x,y)$ defined by (\ref{1:close}) is affine equivalent to $F(x)$ defined by \eqref{F}, and Theorem \ref{thm1} can be equivalently expressed as properties of $F(x)$, which we describe as below.

\begin{thm}\label{thm2}
Let $n=2m$, $m$ odd, $\gcd(n,k)=1$, $\theta\in\ftwom^*$ and $F(x)$ be the polynomial defined by \eqref{F} and \eqref{ci-theta}. Then $F(x)$ is a permutation on $\ftwon$ with $\beta(F)=4$.
\end{thm}

\subsection{Discussion: $F(x+a)+F(x)=b$}

Now to prove our main results, i.e.,  Theorem \ref{thm1} and equivalently Theorem \ref{thm2}, we first study for any $a\in\ftwon^*$, $b\in\ftwon$ the equation
\begin{equation}\label{3:equation0}
F(x+a)+F(x)=b.
\end{equation}
Here $F(x)$ is given by \eqref{F}. Denote
\begin{eqnarray} \label{3:ha}
H_a(x):=F(x+a)+F(x)+F(a).
\end{eqnarray}
Since $F(x)$ is a quadratic polynomial, we have
\begin{equation*}
H_a(x)=\tau_{1}'\overline{x}^{2^k}+\tau_{2}'x^{2^k}+\tau_{3}'\overline{x}+\tau_{4}'x,
\end{equation*}
\begin{eqnarray*}
\begin{array}{llllll}
  \tau_{1}'&=&c_{2}\overline{a}+c_{4}a, &\tau_{2}'&=&c_{1}a +c_{3}\overline{a}, \\
 \tau_{3}'&=&c_{2}\overline{a}^{2^k}+c_{3}a^{2^k},  &\tau_{4}'&=&c_{1}a^{2^k} + c_{4}\overline{a}^{2^k}.
\end{array}
\end{eqnarray*}
Equation \eqref{3:equation0} becomes
\begin{eqnarray} \label{3:hab} H_a(x)=F(a)+b.\end{eqnarray}
Substituting $x$ with $ax$, the above equation becomes
\begin{equation} \label{equation1}
\tau_{1}\overline{x}^{2^k}+\tau_{2}x^{2^k}+\tau_{3}\overline{x}+\tau_{4}x+\tau_{5}= 0
\end{equation}
where $\tau_{5}=F(a)+b$ and other $\tau_i$'s are given by
\begin{eqnarray*}
\begin{array}{llllll}
  \tau_{1}&=&c_{2}\overline{a}^{2^k+1}+c_{4}a\overline{a}^{2^k}, &\tau_{2}&=&c_{1}a^{2^k+1} +c_{3}a^{2^k}\overline{a}, \\
 \tau_{3}&=&c_{2}\overline{a}^{2^k+1}+c_{3}a^{2^k}\overline{a},  &\tau_{4}&=&c_{1}a^{2^k+1} + c_{4}a\overline{a}^{2^k}.
\end{array}
\end{eqnarray*}
Taking $2^m$-th power on both sides of \eqref{equation1} gives
\begin{eqnarray} \label{equation1-1}
\overline{\tau}_{1}x^{2^k}+\overline{\tau}_{2}\overline{x}^{2^k}+
\overline{\tau}_{3}x+\overline{\tau}_{4}\overline{x}+\overline{\tau}_{5}= 0,
\end{eqnarray}
then by $\overline{\tau}_{2}\cdot \eqref{equation1}+\tau_{1}\cdot \eqref{equation1-1}$ one has
\begin{equation} \label{equation2}
v_{1}x^{2^k}+v_{2}\overline{x}+v_{3}x+v_{4} = 0,
\end{equation}
 where
\begin{eqnarray*}
\begin{array}{llllll}
 v_{1}&=&\tau_{1}\overline{\tau}_{1}+\tau_{2}\overline{\tau}_{2}, &v_2&=&\tau_{1}\overline{\tau}_{4}+\overline{\tau}_{2}\tau_{3}, \\
 v_{3}&=&\tau_{1}\overline{\tau}_{3}+\overline{\tau}_{2}\tau_{4},  &v_4&=&\tau_{1}\overline{\tau}_{5}+\overline{\tau}_{2}\tau_{5}.
\end{array}
\end{eqnarray*}
It is easy to verify that the $v_i$'s and $\tau_i$'s satisfy the following properties:
\begin{itemize}
\item[(i)] $v_1+v_2+v_3=0$;

\item[(ii)] $\tau_1+\tau_2=\tau_3+\tau_4=\tau_5+b$;

\item[(iii)] $v_{4}=v_{1}+\tau_{1}\overline{b}+\overline{\tau}_{2}b$;

\item[(iv)] $\tau_{1}\overline{v}_{3}+\tau_{2}v_{2}+\tau_{3}v_{1}=
    \tau_{1}\overline{v}_{2}+\tau_{2}v_{3}+\tau_{4}v_{1}=\tau_{1}\overline{v}_{4}+\tau_{2}v_{4}+\tau_{5}v_{1}=0$.
\end{itemize}
Hence if $v_1\ne 0$, we can write \eqref{equation2} as
\begin{equation} \label{differential-equation}
x^{2^k}+\frac{v_{2}}{v_{1}}\overline{x}+(1+\frac{v_{2}}{v_{1}})x+\frac{\tau_{1}\overline{b}+
\overline{\tau}_{2}b}{v_{1}} +1=0.
\end{equation}

\begin{lem} \label{lem-4}
If $v_{1}\neq 0$, then \eqref{equation1}  and  \eqref{differential-equation} have the same set of solutions in $\ftwon$.
\end{lem}

\begin{proof}
It suffices to show that \eqref{equation1} can be derived from \eqref{differential-equation}. Noting that $v_1=\overline{v}_1$, by using \eqref{equation2}, we obtain
\begin{equation*}
x^{2^k}=\frac{v_{2}\overline{x}+v_{3}x+v_{4}}{v_{1}} {\;\;\rm and\;\;} \overline{x}^{2^k}=\frac{\overline{v}_{2}x+\overline{v}_{3}\overline{x}+\overline{v}_{4}}{v_{1}}.
\end{equation*}
Then we can compute
\begin{eqnarray*}
&&\tau_{1}\overline{x}^{2^k}+\tau_{2}x^{2^k}+\tau_{3}\overline{x}+\tau_{4}x+\tau_{5}\\
&=&\frac{\tau_{1}}{v_{1}}(\overline{v}_{2}x+\overline{v}_{3}\overline{x}+\overline{v}_{4})
 +\frac{\tau_{2}}{k_{1}}(v_{2}\overline{x}+v_{3}x+v_{4})
 +(\tau_{3}\overline{x}+\tau_{4}x+\tau_{5})\\
&=&\frac{1}{v_{1}}\left[\overline{x}(\tau_{1}\overline{v}_{3}+\tau_{2}v_{2}+\tau_{3}v_{1})
 +x(\tau_{1}\overline{v}_{2}+\tau_{2}v_{3}+\tau_{4}v_{1})
 +(\tau_{1}\overline{v}_{4}+\tau_{2}v_{4}+\tau_{5}v_{1})\right] \\
&=&0,
\end{eqnarray*}
which is \eqref{equation1}. This completes the proof of Lemma \ref{lem-4}.
\end{proof}
Indeed we claim that
\begin{lem} \label{lem-v1}
$v_1 \ne 0$ for any $\theta \in \ftwom^*$ and $a \in \ftwon^*$.
\end{lem}
\begin{proof}
By using the values of $c_i$'s in \eqref{ci-theta}, we can compute that
\begin{eqnarray*}
v_{1}&=&\tau_{1}\overline{\tau}_{1}+\tau_{2}\overline{\tau}_{2}
\\&=&(a\overline{a})^{2^k}
\left[(c_{1}c_{3}+c_{2}c_{4})(a^{2}+\overline{a}^{2})+(c_1^2+c_{2}^2+c_{3}^2+c_{4}^2)a\overline{a}\right],
\\&=&(a\overline{a})^{2^k+1}(\theta^2)^{2^k}(1+\theta+\theta^{2^k})^{2}
\left[(\theta+1)(\overline{\gamma}+\gamma)+\theta^2\right],
\end{eqnarray*}
where $\gamma=\overline{a}/a$. Since $m$ is odd, $1+\theta+\theta^{2^k}\ne 0$. Suppose $(\theta+1)(\overline{\gamma}+\gamma)+\theta^2= 0$. Obviously $\theta \ne 1$. Since $\overline{\gamma}=a/\overline{a}=1/\gamma$, we have
\begin{eqnarray} \label{3:theta} \gamma^2+\frac{\theta^2}{\theta+1} \, \gamma+1=0.\end{eqnarray}
Noting that
\[\Tr^m_1 \left(\frac{(\theta+1)^2}{\theta^4}\right)=\Tr_1^m\left(\frac{1}{\theta^2}+\frac{1}{\theta^4}\right)=0,\]
by Lemma \ref{lem0-0}, Equation (\ref{3:theta}) is solvable in $\ftwom$, that is $\gamma=\overline{a}/a \in \ftwom^*$. However, the relation $\gamma=\overline{\gamma}$ implies that
\[a^2=\overline{a}^2 \Longrightarrow a=\overline{a} \ne 0,\]
hence $\gamma=1$, and $(\theta+1)(\overline{\gamma}+\gamma)+\theta^2=\theta^2=0$, contradiction to the fact that $\theta \ne 0$. Thus $(\theta+1)(\overline{\gamma}+\gamma)+\theta^2 \ne 0$ and hence $v_1\ne 0$ for any $a \in \ftwon^*$.
\end{proof}

Returning to Equation \eqref{differential-equation} and comparing it with Lemma \ref{lemma-core}, we have

\begin{lem} \label{3:gen-eq}
For Equation \eqref{differential-equation}, denote
\[\mu=\frac{v_2}{v_1}, \qquad \gamma=\frac{\overline{a}}{a}. \]
Let $\xi \in\ftwom$ and $\lambda\in\ftwon$ be defined by the equations
\begin{eqnarray*} 
\xi^{2^k-1}=1+\mu+\overline{\mu},\;\;\; \lambda^{2^k}+\lambda=\mu \xi.
\end{eqnarray*}
Then we have
\begin{eqnarray*} \label{3:eqgen} 1 +\mu+\overline{\mu} \ne 0, \quad \mbox{ and } \quad \lambda+\overline{\lambda}=\xi.
\end{eqnarray*}
Hence by Lemma \ref{lemma-core}, Equation \eqref{differential-equation} has either 0 or 4 solutions for any $a \in \ftwon^*$ and any $b \in \ftwon$. Moreover
\begin{eqnarray} \label{3:mu}
\mu&=& \frac{\theta^2(1+\theta^{2^k})\gamma^{2^k}+(\theta^2)^{2^k}(\theta+1)\gamma+((\theta+1)^{2^k+1}+1)^{2}}
{(\theta^2)^{2^k}[(\theta+1)(\overline{\gamma}+\gamma)+\theta^2]}, \\
\label{3:xi}
\xi&=&\frac{(\theta+1)(\overline{\gamma}+\gamma)+\theta^2}{\theta^2},
\end{eqnarray}
and $\lambda$ can be taken as
\begin{eqnarray}
\label{3:lam}
\lambda &=&\frac{(1+\theta)\overline{a}}{a\theta^2}+\frac{1}{\theta^2}+\omega.
\end{eqnarray}

\end{lem}

\begin{proof}
All the above facts can be checked easily with some computation. First,
\begin{align}
\mu=&\frac{v_{2}}{v_{1}}=\frac{(c_{1}c_{4}+c_{2}c_{3})\gamma^{2^k}+(c_{1}c_{3}+c_{2}c_{4})\gamma+(c_{3}+c_{4})^{2}}{(c_{1}c_{3}+c_{2}c_{4})(\gamma+
\overline{\gamma})+(c_{1}+c_{2}+c_{3}+c_{4})^{2}}.  \nonumber
\end{align}
Then using the values of $c_i$'s in \eqref{ci-theta}, Equation \eqref{3:mu} can be easily verified. Second, using the value of $\mu$ in \eqref{3:mu}, one can obtain
\begin{eqnarray*} \label{equation-1+zeta}
1+\mu+\overline{\mu}
&=&1+\frac{\theta^2(1+\theta^{2^k})(\overline{\gamma}+\gamma)^{2^k}+(\theta^2)^{2^k}(\theta+1)(\overline{\gamma}+\gamma)}
{(\theta^2)^{2^k}[(\theta+1)(\overline{\gamma}+\gamma)+\theta^2]} \nonumber \\
&=&\frac{\theta^2(1+\theta^{2^k})(\overline{\gamma}+\gamma)^{2^k}+(\theta^2)^{2^k}\theta^2}
{(\theta^2)^{2^k}[(\theta+1)(\overline{\gamma}+\gamma)+\theta^2]}  \nonumber\\
&=&\left(\frac{(\theta+1)(\overline{\gamma}+\gamma)+\theta^2}{\theta^2}\right)^{2^k-1}.
\end{eqnarray*}
Now Equation \eqref{3:xi} is clear due to the fact that $\gcd(2^k-1,2^n-1)=1$. Third, one has
\begin{eqnarray*}\label{muxi}
\mu \xi&=&\frac{\theta^2(1+\theta^{2^k})\gamma^{2^k}+(\theta^2)^{2^k}(\theta+1)\gamma+((\theta+1)^{2^k+1}+1)^{2}}
{(\theta^2)^{2^k}[(\theta+1)(\overline{\gamma}+\gamma)+\theta^2]} \frac{(\theta+1)(\overline{\gamma}+\gamma)+\theta^2}{\theta^2} \nonumber\\
&=&\left(\frac{(1+\theta)\gamma}{\theta^2}+\frac{1}{\theta^2}\right)^{2^k}+
\left(\frac{(1+\theta)\gamma}{\theta^2}+\frac{1}{\theta^2}\right)+1.
\end{eqnarray*}
Since $k$ is odd, we have $1=\omega^{2^k}+\omega$, the value of $\lambda$ given by Equation \eqref{3:lam} is a solution to the equation $\lambda^{2^k}+\lambda=\mu \xi$. Using this value of $\lambda$, one can easily verify that $\lambda+\overline{\lambda}=\xi$. This completes the proof of Lemma \ref{3:gen-eq}.

\end{proof}

\section{Proof of Theorem \ref{thm2}} \label{sec4}

\subsection{Permutation and differential uniformity}

For any $a\in\ftwon^*$ and $b\in\ftwon$, consider the equation
\begin{equation}\label{4:equation0}
F(x+a)+F(x)=b,
\end{equation}
where $F(x)$ is given by \eqref{F}. It is known that Equation \eqref{4:equation0} is equivalent to Equation \eqref{differential-equation}, which by Lemma \ref{3:gen-eq}, always has either 0 or 4 solutions in $\ftwon$. Hence $F$ is differential $4$-uniform.

Now for any $a \in \ftwon^*$, consider the equation
\begin{equation}\label{4:equation01}
F(x+a)+F(x)=0,
\end{equation}
that is, $b=0$ in \eqref{4:equation0}. Accordingly, this is equivalent to Equation \eqref{differential-equation} with $b=0$, which can be written explicitly as
\begin{eqnarray*} \label{4:eqper} x^{2^k}+\mu \overline{x}+(1+\mu)x+1=0,\end{eqnarray*}
where $\mu$ is defined in \eqref{3:mu}. According to Lemma \ref{3:gen-eq}, since $m$ is odd and $\overline{\lambda}+\lambda=\xi$, we have
\[\Tr_1^n \left(\frac{\lambda^{2^k}}{\xi^{2^k}}\right)=\Tr_1^n \left(\frac{\lambda}{\xi}\right)=\Tr_1^m \left(\frac{\overline{\lambda}+\lambda}{\xi}\right)=\Tr_1^m(1)=1,\]
hence by Lemma \ref{lemma-core}, Equation \eqref{differential-equation} with $b=0$ and equivalently Equation \eqref{4:equation01} is not solvable in $\ftwon$ for any $a \in \ftwon^*$. So we concludes that $F$ is a permutation.

\subsection{Boomerang uniformity}
To compute the boomerang uniformity of $F$, we need considerably more effort.

First recall a new formulation of the boomerang uniformity of $F(x)$ in \cite{LQSL}, which allow us to compute $\beta(F)$ conveniently without using the compositional inverse $F^{-1}$:
\begin{lem}\label{LQSL}
Let $F(x)$ be a permutation over $\ftwon$. Denote by $S_{F}(a,b)$ the number of solutions $(x,y)\in \ftwon^2$ of the equation system
\begin{numcases}{}
F(x+a)+F(y+a)=b, \nonumber \\
F(x)+F(y)=b. \nonumber
\end{numcases}
Then
\[\beta(F)=\max \left\{S_F(a,b): a,b\in \ftwon^*\right\}. \]
\end{lem}
Since $\beta(F) \ge \delta(F)=4$, to complete the proof of Theorem \ref{thm2}, it suffices to show that $S_{F}(a,b)\leq 4$ for any $a,b \in \ftwon^*$. Now for any fixed $a,b \in \ftwon^*$, the value $S_{F}(a,b)$ is equal to the number of solutions $(x,y) \in \ftwon^2$ of the following equation system
\begin{numcases}{}
F(x+a)+F(x)+F(y+a)+F(y)=0, \label{boomerang-diff-1} \\
F(x)+F(y)=b. \label{boomerang-diff-2}
\end{numcases}
Since $b \ne 0$, obviously $x+y \ne 0$.

We first consider Equation \eqref{boomerang-diff-1}. Using the function $H_a(x)$ defined in \eqref{3:ha} which is linear in both $a$ and $x$, Equation \eqref{boomerang-diff-1} can be rewritten as
\[H_a(x)+H_a(y)=H_a(x+y)=0. \]
Letting $z=x+y \in \ftwon^*$, tracing back to Equation \eqref{3:hab} with the right hand being 0, the above equation has roots $ax$ where $x$ satisfies the equation
\[x^{2^k}+ \mu \overline{x}+(1+\mu)x=0,\]
and $\mu$ is given by Equation \eqref{3:mu} in Lemma \ref{3:gen-eq}. Using Lemma \ref{3:gen-eq} and Lemma \ref{lemma-core}, we conclude that $z=ax \in Z_a$ where the set $Z_a$ is given by
\begin{eqnarray*} \label{4:za} Z_a:=\left\{a,\eta_a, a+\eta_a\right\},\end{eqnarray*}
and
\begin{eqnarray}\label{eta}
\eta_a=a\lambda=\frac{(1+\theta)\overline{a}}{\theta^2}+\left(\frac{1}{\theta^2}+w\right)a.
\end{eqnarray}

Next, we consider \eqref{boomerang-diff-2}. Using $y=x+z$, Equation \eqref{boomerang-diff-2} becomes
\begin{eqnarray} \label{4:hz} H_z(x)=F(z)+b. \end{eqnarray}
It is known from Lemma \ref{3:gen-eq} that the above equation has at most four solutions in $\ftwon$ for each $z \in Z_a$, so immediately we obtain $\beta(F) \le 12$. To find the exact value of $\beta(F)$, we need to consider more carefully the solvability of \eqref{4:hz} for $z \in Z_a$.

Using the equivalence between \eqref{3:hab} and \eqref{differential-equation} and applying Lemma \ref{3:gen-eq}, we conclude that for any $z\in Z_a$, Equation \eqref{4:hz} is equivalent to
\begin{align}  \label{boomerang-2}
x^{2^k}+\mu_z\overline{x}+(1+\mu_z)x+\nu_z= 0,
\end{align}
where $\mu_z$ and $\nu_z$ are given by
\begin{eqnarray}\label{boomerang-nu}
  \mu_z &=& \frac{\theta^2(1+\theta^{2^k})\gamma^{2^k}_z+(\theta^2)^{2^k}(\theta+1)\gamma_z+((\theta+1)^{2^k+1}+1)^{2}}
{(\theta^2)^{2^k}\left[(\theta+1)(\overline{\gamma}_z+\gamma_z)+\theta^2\right]}, \nonumber \\
   \nu_z &=& 1+\frac{1}{z^{2^k}}
\frac{(c_{1}\overline{z}+c_{3}z)b+(c_{2}\overline{z}+c_{4}z)\overline{b}}
{(\theta^2)^{2^k}(1+\theta+\theta^{2^k})^{2}\left[(\theta+1)(z^{2}+
\overline{z}^{2})+\theta^2z\overline{z}\right]}. \label{4:vz}
\end{eqnarray}
Here $\gamma_z=\overline{z}/z$. Further, for $\xi_z$ satisfying $\xi^{2^k-1}_z=1+\mu_z+\overline{\mu}_z$ and $\lambda_z$ satisfying $\lambda_z^{2^k}+\lambda_z=\mu_z \xi_z$, we have
\begin{eqnarray} \label{4:xi} \xi_z &=&\frac{(\theta+1)(\overline{\gamma}_z+\gamma_z)+\theta^2}{\theta^2},\\
\label{4:lam} \lambda_z &=&\frac{(1+\theta)\gamma_z}{\theta^2}+\frac{1}{\theta^2}+w.\end{eqnarray}
Thus, by Lemma \ref{3:gen-eq} and by Lemma \ref{lemma-core}, Equation \eqref{boomerang-2} (and equivalently \eqref{4:hz}) has either $0$ or $4$ solutions in $\ftwon$, and it has $4$ solutions in $\ftwon$ if and only if
\begin{eqnarray} \label{4:4roots} \Tr_{1}^{m}\left(\Delta_{z}\right)=0 \;\;{\rm and}\;\;\Tr_{1}^{n}\left(\frac{\lambda^{2^k}_z\overline{\nu}_z}{\xi^{2^k}_z}\right)=0,\end{eqnarray}
where
\begin{eqnarray*} \label{4:dz} \Delta_{z}=\frac{\nu_z+\overline{\nu}_z}{\xi^{2^k}_z}. \end{eqnarray*}
According to \eqref{boomerang-nu}, it can be readily verified that
\begin{eqnarray*}
&&(\overline{\nu}_z+ \nu_z)\left[(\theta^2)^{2^k}(1+\theta+\theta^{2^k})^{2}\left[
(\theta+1)(z^{2}+\overline{z}^{2})+\theta^2z\overline{z}\right]\right]
\\&=& \left(\frac{c_{2}z+c_{4}\overline{z}}{\overline{z}^{2^k}}+
\frac{c_{1}\overline{z}+c_{3}z}{z^{2^k}}\right)b+ \left(\frac{c_{2}\overline{z}+c_{4}z}{z^{2^k}}+\frac{c_{1}z+c_{3}
\overline{z}}{\overline{z}^{2^k}}\right)\overline{b}\\
 &=&\Tr_{m}^{n}\left(\left(\frac{c_{2}\overline{z}+c_{4}z}{z^{2^k}}+\frac{c_{1}z+
 c_{3}\overline{z}}{\overline{z}^{2^k}}\right)\overline{b}\right),
\end{eqnarray*}
which leads to
\begin{eqnarray} 
\Tr_{1}^{m}(\Delta_{z})
&=&\Tr_{1}^{m}\left(\frac{(\theta^2)^{2^k}(\nu_z+\overline{\nu}_z)}{((\theta+1)
(\overline{\gamma}_z+\gamma_z)+\theta^2)^{2^k}}\right) \nonumber\\
&=&\Tr_{1}^{m}\left(\frac{(z\overline{z})^{2^k}(\theta^2)^{2^k}(\nu_z+\overline{\nu}_z)}
{((\theta+1)(z^{2}+\overline{z}^{2})+\theta^2z\overline{z})^{2^k}}\right)  \nonumber \\
&=&\Tr_{1}^{n}\left(\frac{\left(\overline{z}^{2^k}(c_{2}\overline{z}+c_{4}z)
+z^{2^k}(c_{1}z+c_{3}\overline{z})\right)\overline{b}}
{(1+\theta+\theta^{2^k})^{2}((\theta+1)(z^{2}+\overline{z}^{2})+\theta^2z\overline{z})^{2^k+1}}\right).  \nonumber
\end{eqnarray}
Denoting
\begin{eqnarray*} \label{4:ez}
E(z):=(\theta+1)\left(z^{2}+\overline{z}^{2}\right)+\theta^2z\overline{z},
\end{eqnarray*}
we have
\begin{eqnarray}\label{Delta}
\Tr_{1}^{m}(\Delta_{z})
&=&\Tr_{1}^{n}\left(\frac{F(z)\overline{b}}
{(1+\theta+\theta^{2^k})^{2}E(z)^{2^k+1}}\right).
\end{eqnarray}

\begin{lem} \label{lem-invariant}
For any $a \in \ftwon^*$, we have
\begin{enumerate}
\item [(1)] $\sum_{z \in Z_a}F(z)=0$; and
 \item [(2)] $E(z) \ne 0 $ and is an invariant for any $z\in Z_a$.
\end{enumerate}
Here $Z_a=\left\{a,\eta_a,a+\eta_a\right\}$ and $\eta_a$ is defined by \eqref{eta}.
\end{lem}
\begin{proof}
See Appendix B.
\end{proof}
Lemma \ref{lem-invariant} indicates that $\sum_{z\in Z_a}\Tr_{1}^{m}(\Delta_{z})=0$ which implies that $\Tr_{1}^{m}(\Delta_{z})=0$ for exactly one $z\in Z_a$ or for all $z\in Z_a$. If $\Tr_{1}^{m}(\Delta_{z})=0$ for exactly one $z\in Z_a$, then $S_{F}(a,b)\leq 4$ by Lemma \ref{lemma-core}.

Now let us assume that $\Tr_{1}^{m}(\Delta_{z})=0$ for all $z\in Z_a$. For simplicity, define
\[\phi(z):=\frac{(1+\theta)\overline{z}}{\theta^2}+\left(\frac{1}{\theta^2}+\omega \right)z.\]
Using values $\nu_z$, $\xi_z$ and $\lambda_z$ given in \eqref{4:vz}, \eqref{4:xi} and \eqref{4:lam} respectively and noting that $\lambda_z+\overline{\lambda}_z=\xi_z$, we obtain
\begin{eqnarray} \label{V}
 \Tr_{1}^{n}\left(\frac{\lambda_z^{2^k}\overline{\nu}_z}{\xi_z^{2^k}}\right)
&=&
\Tr_{1}^{n}\left(\frac{\left(\overline{\phi(z)}^{2^k}(c_{2}\overline{z}+c_{4}z)+\phi(z)^{2^k}(c_{1}z+c_{3}\overline{z})\right)\overline{b}}
{(1+\theta+\theta^{2^k})^{2}\left((\theta+1)(z^{2}+\overline{z}^{2})+\theta^2z\overline{z}\right)^{2^k+1}}\right)+\Tr_1^n\left(\frac{\lambda_z}{\xi_z}\right) \nonumber \\
&=&
\Tr_{1}^{n}\left(\frac{H(z)\overline{b}}{(1+\theta+\theta^{2^k})^{2}E(z)^{2^k+1}}\right)  +\Tr_1^m\left(\frac{\lambda_z+\overline{\lambda}_z}{\xi_z}\right)  \nonumber \\
&=& \Tr_{1}^{n}\left(\frac{H(z)\overline{b}}{(1+\theta+\theta^{2^k})^{2}E(z)^{2^k+1}}\right)+1,
\end{eqnarray}
where $E(z)$ is defined as in Lemma \ref{lem-invariant} and $$H(z):=\overline{\phi(z)}^{2^k}(c_{2}\overline{z}+c_{4}z)+\phi(z)^{2^k}(c_{1}z+c_{3}\overline{z}).$$

\begin{lem} \label{lem-Hz}
Let $H(z)$ be defined as above and $\eta_a$ be defined by \eqref{eta}, then we have
$$H(a)=F(a+\eta_a),\;H(\eta_a)=F(a),\; H(a+\eta_a)=F(\eta_a).$$
\end{lem}
\begin{proof}
See Appendix C.
\end{proof}

Combining Lemma \ref{lem-invariant}, Lemma \ref{lem-Hz}, \eqref{Delta} and \eqref{V}, we see that if
$\Tr_{1}^{m}(\Delta_{z})=0$ for all $z\in Z_a$, then
\begin{eqnarray*}
 \Tr_{1}^{n}\left(\frac{\lambda_z^{2^k}\overline{\nu}_z}{\xi_z^{2^k}}\right)
= \Tr_{1}^{n}\left(\frac{H(z)\overline{b}}{(1+\theta+\theta^{2^k})^{2}E(a)^{2^k+1}}\right)+1=1, \quad \forall\; z \in Z_a.
\end{eqnarray*}
This means that Equation \eqref{4:4roots} never holds for $z \in Z_a$, that is, $S_{F}(a,b)=0$.  Combining these two cases we conclude that $S_{F}(a,b)\leq 4$ for any $a,b \in \ftwon^*$. Hence $\beta(F) \le 4$. This completes the proof of Theorem \ref{thm2} and equivalently Theorem \ref{thm1}.

\section*{Acknowledgements}

 The authors would like to thank Dr. Chunming Tang for helpful discussions. This work was supported by the National Natural Science Foundation of China (Nos. 61702166, 61761166010) and by the Research Grants Council (RGC) of Hong Kong (Nos. N\_HKUST169/17).

\section*{Appendix A: Proof of Lemma \ref{lemma-core}}

 \setcounter{equation}{0}
\renewcommand\theequation{A.\arabic{equation}}

\begin{proof}

Let \begin{equation*} \label{2:equation0} z=x+\overline{x}.\end{equation*}
Then the equation $L_{\mu,\nu}(x)=0$ becomes
\begin{equation}
\label{2:equation4} x^{2^k}+x+\mu z+\nu=0.
\end{equation}
Taking $2^m$-th power on both sides of \eqref{2:equation4} and adding them together gives
\begin{equation} \label{2:equation5}
z^{2^k}+(1+\mu+\overline{\mu})z+\nu+\overline{\nu}=0.
\end{equation}
Taking $2^k$-th power consecutively on both sides of \eqref{2:equation4}, one can also obtain
\begin{equation*}
x+x^{2^{km}}=x+\overline{x}=\sum_{i=0}^{m-1}(\mu z+\nu)^{2^{ki}}.
\end{equation*}
Hence solving $L_{\mu,\nu}(x)=0$ for $x \in \ftwon$ is equivalent to solving the system of equations (\ref{2:equation4}), (\ref{2:equation5}) and
\begin{equation}\label{2:equation6}
\sum_{i=0}^{m-1}(\mu z+\nu)^{2^{ki}}+z=0
\end{equation}
for $x \in \ftwon$ and $z \in \ftwom$. Note that $\sum_{i=0}^{m-1}(\mu z+\nu)^{2^{ki}}+z\in \mathbb{F}_2$.

Without checking the solvability of \eqref{2:equation6}, since $\gcd(n,k)=1$, for any $\mu$ and $\nu$, Equation \eqref{2:equation5} has at most two solutions for $z \in \ftwom$, and for each such $z$, Equation \eqref{2:equation4} has at most two solutions for $x \in \ftwon$, hence the equation $L_{\mu,\nu}(x)=0$ has at most 4 solutions. Also observe that whenever $z \in \ftwom$ is a solution to Equation \eqref{2:equation5} that satisfies Equation \eqref{2:equation6}, one always has
\[\Tr_1^n\left(\mu z+\nu\right)=\Tr_1^m\left(\mu z+\nu+\overline{\mu z+\nu}\right)=z+\overline{z}=0,\]
hence for such $z$, by Lemma \ref{lem0}, Equation \eqref{2:equation4} is always solvable with two solutions $x \in \ftwon$. We conclude that the number of solutions of $L_{\mu,\nu}(x)=0$ equals two times the number of $z \in \ftwon$ satisfying \eqref{2:equation5} and \eqref{2:equation6}.

Now we study in more details the solvability of \eqref{2:equation5} and \eqref{2:equation6}.

\noindent {\bf Case 1}: $1+\mu+\overline{\mu}=0$.

In this case, \eqref{2:equation5} has a unique solution $z$ such that $z^{2^k}=\nu+\overline{\nu}$, and \eqref{2:equation6} is equivalent to
\begin{equation*}
\sum_{i=0}^{m-1}\left(\mu^{2^k} z^{2^k}+\nu^{2^k}\right)^{2^{ki}}=z^{2^k},
\end{equation*}
and this proves (i) of Lemma \ref{lemma-core}.

\noindent {\bf Case 2}: $1+\mu+\overline{\mu}\ne 0$.

Let $\xi$, $\Delta$ be defined by \eqref{notation} and $z=\xi\rho$, then \eqref{2:equation5} becomes
\begin{equation} \label{2:rho}
\rho^{2^k}+\rho=\Delta
\end{equation}
which has solutions for $\rho \in \ftwom$ if and only if $\Tr_1^m(\Delta)=0$.

We now assume that $\Tr_1^m(\Delta)=0$. The two solutions $z_1,z_2 \in \ftwom$ to Equation \eqref{2:equation5} satisfy the relation
\[z_1+z_2=\xi.\]
Using $\lambda^{2^k}+\lambda=\mu \xi$, we have
\begin{eqnarray*} \sum_{j=1}^2 \left(\sum_{i=0}^{m-1} \left(\mu z_{j}+\nu\right)^{2^{ki}}+z_{j}\right) = \sum_{i=0}^{m-1}\left(\mu \xi \right)^{2^{ki}}+\xi
=\lambda+\overline{\lambda}+\xi\in \mathbb{F}_2. \end{eqnarray*}
It is easy to see that if $\lambda+\overline{\lambda}=\xi+1$, then among $z_1$ and $z_2$, exactly one element satisfies Equation \eqref{2:equation6}, hence the equation $L_{\mu,\nu}(x)=0$ has two solutions. This proves (ii) of Lemma \ref{lemma-core}.

Finally, let us assume $\overline{\lambda}+\lambda=\xi$. In this case, either both $z_1$ and $z_2$  satisfy \eqref{2:equation6} or neither satisfy \eqref{2:equation6}, hence the equation $L_{\mu,\nu}(x)=0$ has either $4$ or $0$ solution.

Let $z=\xi\rho$ be a solution to Equation \eqref{2:equation5} where $\rho \in \ftwom$ satisfies (\ref{2:rho}). We will compute directly the left hand side of Equation \eqref{2:equation6} for $z$. For this purpose denote
\[h(z):= \sum_{i=0}^{m-1} \left(\mu \xi \rho+\nu\right)^{2^{ki}}+\xi \rho.\]
Using (\ref{2:rho}) and the relation $\sum_{i=0}^{m-1}\left(\mu \xi \right)^{2^{ki}}=\lambda+\overline{\lambda}=\xi$ we can obtain
\begin{eqnarray}\label{sumxirho}
\sum_{i=0}^{m-1}(\mu\xi\rho)^{2^{ki}}
&=&\sum_{i=1}^{m-1}(\mu\xi)^{2^{ki}}\rho^{2^{ki}}+\mu\xi\rho
=\sum_{i=1}^{m-1}(\mu\xi)^{2^{ki}}\left(\rho+\sum_{j=0}^{i-1}\Delta^{2^{kj}} \right)+\mu\xi\rho \nonumber\\
&=&\rho\sum_{i=0}^{m-1}(\mu\xi)^{2^{ki}}+\sum_{i=1}^{m-1}(\mu\xi)^{2^{ki}}\sum_{j=0}^{i-1}\Delta^{2^{kj}}\nonumber\\
&=&\rho\xi+\sum_{i=1}^{m-1}(\mu\xi)^{2^{ki}}\sum_{j=0}^{i-1}\Delta^{2^{kj}}.
\end{eqnarray}
As for the second term on the right side of (\ref{sumxirho}), using $\Tr_1^m(\Delta)=\sum_{i=0}^{m-1}\Delta^{2^{ki}}=0$, one can obtain
\begin{eqnarray} \label{lamDelta}
\sum_{i=1}^{m-1}\sum_{j=0}^{i-1}(\mu\xi)^{2^{ki}}\Delta^{2^{kj}}
&=&\sum_{j=0}^{m-2}\Delta^{2^{kj}}\sum_{i=j+1}^{m-1}(\mu\xi)^{2^{ki}}
=\sum_{j=0}^{m-2}\Delta^{2^{kj}}\sum_{i=j+1}^{m-1}(\lambda^{2^k}+\lambda)^{2^{ki}}\nonumber\\
&=&\sum_{j=0}^{m-2}\Delta^{2^{kj}}(\lambda^{2^{k(j+1)}}+\lambda^{2^{km}})\nonumber\\
&=&\sum_{j=0}^{m-2}(\lambda^{2^k}\Delta)^{2^{kj}}+\Delta^{2^{k(m-1)}}\lambda^{2^{km}}\nonumber\\
&=&\sum_{j=0}^{m-1}(\lambda^{2^k}\Delta)^{2^{kj}}.
\end{eqnarray}
Combining (\ref{sumxirho}) and (\ref{lamDelta}) we can easily find
\[h(z)=\sum_{i=0}^{m-1} \left(\lambda^{2^k}\Delta+\nu\right)^{2^{ki}}.\]
Finally, noting that
\[\lambda^{2^k}\Delta+\nu=\frac{\lambda^{2^k}}{\xi^{2^k}}(\overline{\nu}+\nu)+\nu
=\frac{\lambda^{2^k}\overline{\nu}+\overline{\lambda}^{2^k} \nu}{\xi^{2^k}}, \]
we conclude that
\[h(z)= \Tr_1^n \left(\frac{\lambda^{2^k}\overline{\nu}}{\xi^{2^k}}\right).\]
Hence a solution $z$ of \eqref{2:equation5} satisfies \eqref{2:equation6} if and only if
\begin{equation*}
\Tr_{1}^{n} \left(\frac{\lambda^{2^k}\overline{\nu} }{\xi^{2^k}} \right)=0.
\end{equation*}
Thus the equation $L_{\mu,\nu}(x)=0$ has $4$ solutions. This completes the proof of (2) of Lemma \ref{lemma-core}.

Finally, let us assume that $\nu=0, 1+\mu+\overline{\mu} \ne 0$ and $\lambda+\overline{\lambda}=\xi$. Then \eqref{2:equation5} has two solutions $z_1=0, z_2=\xi$ which both satisfy \eqref{2:equation6}. Returning to \eqref{2:equation4}, the corresponding four roots of $L_{\mu,\nu}(x)=0$ are given by $0,1,\lambda, \lambda+1$. Now Lemma \ref{lemma-core} is proved.
\end{proof}

\section*{Appendix B: Proof of Lemma \ref{lem-invariant}}

 \setcounter{equation}{0}
\renewcommand\theequation{B.\arabic{equation}}

\begin{proof}
It is known that $H_a(z)=0$ for any $z \in Z_a$, that is
\[F(z)+F(z+a)+F(a)=0 \qquad \forall \;z \in Z_a. \]
Taking $z=\eta_a$, the first assertion is proved.

We next show $E(z)=(\theta+1)(z^{2}+\overline{z}^{2})+\theta^2z\overline{z}$ is
invariant for $z\in Z_a$.
Firstly, by \eqref{eta} one gets
\[ \eta_a\overline{a}+\overline{\eta}_a a=\frac{(1+\theta)(\overline{a}^2+a^2)}{\theta^2}+a\overline{a}.\] From this one can easily deduce that $E(a+\eta_a)=E(\eta_a)$. Secondly, again by \eqref{eta}, one obtains
\begin{eqnarray*}
\eta_a+\overline{\eta}_a&=&\frac{a+\overline{a}}{\theta}+\omega a+\omega^2\overline{a},\\
\eta_a \cdot \overline{\eta}_a&=&a\overline{a}+\frac{(\theta+1)(\omega \theta^2+1)a^2}{\theta^4}+\frac{(\theta+1)(\omega^2\theta^2+1)\overline{a}^2}{\theta^4},
\end{eqnarray*}
where $\omega \in\ftwon$ satisfying $\omega^2+\omega+1=0$. Then one can easily verify that
\begin{eqnarray*}
&&E(\eta_a)\\&=&(\theta+1)(\eta_a^{2}+\overline{\eta}_a^{2})+\theta^2\eta_a\overline{\eta}_a\\
  &=& (\theta+1)\left(\frac{a+\overline{a}}{\theta}+\omega a+\omega^2\overline{a}\right)^2+ \theta^2a^2\overline{a}^2
    +\frac{(\theta+1)(\omega \theta^2+1)a^2}{\theta^2}+\frac{(\theta+1)(\omega^2\theta^2+1)\overline{a}^2}{\theta^2} \\
   &=&  (\theta+1)(\overline{a}^2+a^2)+\theta^2a\overline{a}= E(a).
\end{eqnarray*}
Noting that
\[ E(a) =a \overline{a} \left((\theta+1) \left(\gamma+\overline{\gamma}\right) +\theta^2\right),\]
where $\gamma=\overline{a}/a$, it is now clear that $E(a) \ne 0$ from the proof of Lemma \ref{lem-v1}. This completes the proof of Lemma \ref{lem-invariant}.
\end{proof}

\section*{Appendix C: Proof of Lemma \ref{lem-Hz}}

 \setcounter{equation}{0}
\renewcommand\theequation{C.\arabic{equation}}

\begin{proof}

Let $H(z)$ be defined as in Lemma \ref{lem-Hz}, i.e.,
\begin{eqnarray*}
H(z)=\phi(z)^{2^k}(c_{1}z+c_{3}\overline{z})+\overline{\phi(z)}^{2^k}(c_{2}\overline{z}+c_{4}z),
\end{eqnarray*}
and $F(x)$ be defined as \eqref{F} which can also be written as
\begin{eqnarray*}
F(x)=x^{2^k}(c_1x+c_3\overline{x})+\overline{x}^{2^k}(c_2\overline{x}+c_4 x),
\end{eqnarray*}
where $c_i$'s are given by \eqref{ci-theta}. For simplicity, define
\begin{eqnarray*}
\psi(x,y)=\frac{(\theta+1)x}{\theta^2}+\left(\frac{1}{\theta^2}+\omega\right)y.
\end{eqnarray*}
Note that $\phi(z)=\frac{(1+\theta)\overline{z}}{\theta^2}+\left(\frac{1}{\theta^2}+\omega \right)z$ and $\eta_a=\phi(a)$ according to \eqref{eta}. Then by a direct calculation one has
\begin{eqnarray*}
F(\eta_a)&=&\eta_a^{2^k}(c_1\eta_a+c_3\overline{\eta}_a)+\overline{\eta}_a^{2^k}(c_2\overline{\eta}_a+c_4 \eta_a) \\
&=&\phi(a)^{2^k}\left(c_1\phi(a)+c_3\overline{\phi(a)}\right)+\overline{\phi(a)}^{2^k}\left(c_2\overline{\phi(a)}+c_4 \phi(a)\right) \\
&=& s_1a^{2^k+1}+s_2\overline{a}a^{2^k}+s_3\overline{a}^{2^k}a+s_4\overline{a}^{2^k+1},
\end{eqnarray*}
where
\begin{eqnarray*}
 s_{1}
&=&\left(\frac{1}{\theta^2}+\omega \right)^{2^k}\psi(c_3,c_1)+\left(\frac{\theta+1}{\theta^2}\right)^{2^k}\psi(c_2,c_4), \\
s_2
&=&\left(\frac{1}{\theta^2}+\omega\right)^{2^k}\overline{\psi(c_1,c_3)}+\left(\frac{\theta+1}{\theta^2}\right)^{2^k}\overline{\psi(c_4,c_2)}, \\
 s_{3}
&=&\left(\frac{\theta+1}{\theta^2}\right)^{2^k}\psi(c_3,c_1)+\left(\frac{1}{\theta^2}+\omega^2\right)^{2^k}\psi(c_2,c_4),  \\
s_4
&=&\left(\frac{\theta+1}{\theta^2}\right)^{2^k}\overline{\psi(c_1,c_3)}+\left(\frac{1}{\theta^2}+\omega^2\right)^{2^k}\overline{\psi(c_4,c_2)}.
\end{eqnarray*}
On the other hand, one gets
\begin{eqnarray*}
H(a)&=&\phi(a)^{2^k}(c_{1}a+c_{3}\overline{a})+\overline{\phi(a)}^{2^k}(c_{2}\overline{a}+c_{4}a) \\
&=& t_1a^{2^k+1}+t_2\overline{a}a^{2^k}+t_3\overline{a}^{2^k}a+t_4\overline{a}^{2^k+1},
\end{eqnarray*}
where
\begin{eqnarray*}
\begin{array}{llllll}
 t_{1}&=&\left(\frac{1}{\theta^2}+\omega\right)^{2^k}c_1+\left(\frac{\theta+1}{\theta^2}\right)^{2^k}c_4, &t_2&=&\left(\frac{1}{\theta^2}+\omega\right)^{2^k}c_3+\left(\frac{\theta+1}{\theta^2}\right)^{2^k}c_2, \\
 t_{3}&=&\left(\frac{\theta+1}{\theta^2}\right)^{2^k}c_1+\left(\frac{1}{\theta^2}+\omega^2\right)^{2^k}c_4,  &t_4&=&\left(\frac{\theta+1}{\theta^2}\right)^{2^k}c_3+\left(\frac{1}{\theta^2}+\omega^2\right)^{2^k}c_2.
\end{array}
\end{eqnarray*}
A straightforward calculation gives
\begin{eqnarray*}
s_1+t_1&=&\left(\frac{1}{\theta^2}+\omega\right)^{2^k}\left(\psi(c_3,c_1)+c_1\right)+\left(\frac{\theta+1}{\theta^2}\right)^{2^k}\left(\psi(c_2,c_4)+c_4\right)\\
&=& \left(\frac{1}{\theta^2}+\omega\right)^{2^k}\overline{\psi(c_3,c_1)}+\left(\frac{\theta+1}{\theta^2}\right)^{2^k}\overline{\psi(c_2,c_4)}.
\end{eqnarray*}
Using the identities $\omega^2+\omega+1=0$, $\omega^{2^k}=\omega^2$ and the facts $\theta, c_i\in\ftwom$ and denoting
\begin{eqnarray*}
 A&=&(\theta+1)c_3+\left(1+\theta^{2^{k+1}+2}\right)c_1+(\theta+1)^{2^k+1}c_2+(\theta+1)^{2^k}c_4,\\
 B&=&\theta^{2^{k+1}}(\theta+1)c_3+\left(\theta^2+\theta^{2^{k+1}}+\theta^{2^{k+1}+2}\right)c_1+(\theta+1)^{2^k}\theta^2c_4,
\end{eqnarray*}
with some straightforward computation, we can verify that
\begin{eqnarray*}
s_1+t_1&=& \theta^{-2^{k+1}-2}A+B\omega^2 = c_1.
\end{eqnarray*}
Using the values of $c_i$'s from \eqref{ci-theta}, one can similarly verify that
\[ s_i+t_i=c_i, \qquad \forall\; i=2,3,4.\]
This shows that $F(\eta_a)+H(a)=F(a)$, that is, $H(a)=F(a)+F(\eta_a)=F(a+\eta_a)$ by Lemma \ref{lem-invariant}. The other two identities on $H$ and $F$ can be proved in a similar manner. This completes the proof of Lemma \ref{lem-Hz}.

\end{proof}

\end{document}